\documentclass[journal]{IEEEtran}
\usepackage{kantlipsum,widetext}
\usepackage{algorithmic}
\usepackage{graphicx}
\usepackage{textcomp}
\usepackage{filecontents}
\usepackage{lipsum}
\usepackage{setspace}
\usepackage{nicematrix}
\usepackage{mathtools}

\usepackage{txfonts}

\usepackage{tabularx}

\usepackage[numbers,sort&compress]{natbib}
\usepackage[utf8]{inputenc}
\usepackage[english]{babel}
\usepackage{cmap}

\newcommand{\T}{\texttt{T}}
\newcommand{\Sgate}{\texttt{S}}
\newcommand{\X}{\texttt{X}}
\newcommand{\Z}{\texttt{Z}}

\newcommand{\CS}{\texttt{CS}}
\newcommand{\CCZ}{\texttt{CCZ}}
\newcommand{\Hgate}{\texttt{H}}
\newcommand{\CNOT}{\texttt{CNOT}}

\RequirePackage{xcolor}
\usepackage[T1]{fontenc}
\usepackage{amsmath,amssymb}
\usepackage{hyperref}
\hypersetup{colorlinks=true}

\usepackage{amsthm}
\newtheorem{theorem}{Theorem}[section]
\newtheorem{definition}{Definition}[section]
\newtheorem{corollary}{Corollary}[theorem]
\newtheorem{remark}{Remark}[theorem]
\newtheorem{lemma}[theorem]{Lemma}
\newtheorem{conjecture}[theorem]{Conjecture}
\usepackage{eczoo}
\usepackage{tikz}
\usepackage{booktabs,multirow,makecell}

\def\BibTeX{{\rm B\kern-.05em{\sc i\kern-.025em b}\kern-.08em
    T\kern-.1667em\lower.7ex\hbox{E}\kern-.125emX}}
\begin{document}
\title{Transversal Clifford and $T$-gate codes of short length and high distance}
\author{Shubham P. Jain\IEEEmembership{} and Victor V. Albert \IEEEmembership{}

\thanks{S. P. Jain is with the Joint Center for Quantum Information and Computer Science at the University of Maryland, College Park, MD 20740, USA. (e-mail: sjn@umd.edu).}
\thanks{V. V. Albert is with the Joint Center for Quantum Information and Computer Science at the University of Maryland, College Park, MD 20740, USA and the National Institute of Standards and Technology, Gaithersburg, MD 20899, USA. (e-mail: vva@umd.edu).}}

\maketitle

\begin{abstract}

The non-local interactions in several quantum device architectures allow for the realization of more compact quantum encodings while retaining the same degree of protection against noise. Anticipating that short to medium-length codes will soon be realizable, it is important to construct stabilizer codes that, for a given code distance, admit fault-tolerant implementations of logical gates with the fewest number of physical qubits. 
To this aim, we construct three kinds of codes encoding a single logical qubit for distances up to $31$. First, we construct the smallest known doubly even codes, all of which admit a transversal implementation of the Clifford group. Applying a doubling procedure  [\href{https://arxiv.org/abs/1509.03239}{arXiv:1509.03239}] to such codes yields the smallest known weak triply even codes for the same distances and number of encoded qubits. This second family of codes admit a transversal implementation of the logical $\texttt{T}$-gate. Relaxing the triply even property, we obtain our third family of triorthogonal codes with an even lower overhead at the cost of requiring additional Clifford gates to achieve the same logical operation. 
To our knowledge, these are the smallest known triorthogonal codes for their respective distances. While not qLDPC, the stabilizer generator weights of the code families with transversal \T-gates scale roughly as the square root of their lengths.

\end{abstract}


\begin{table}[t]
\caption{Self-dual based triorthogonal codes}
\label{tab:triortho}
\centering
\begin{tabular}{ccc}
\toprule 
self dual  & self dual CSS  & triorthogonal \tabularnewline
\midrule 
$[8,4,4]$ & $[[7,1,3]]$~\cite{steane_multiple_1996}  & $[[15,1,3]]$~\cite{bravyi_doubled_2015} \tabularnewline
\midrule 
 $[18,9,4]$ & $[[17,1,5]]$~\cite{bombin_topological_2006}  & $[[49,1,5]]$~\cite{bravyi_doubled_2015}\tabularnewline
\midrule 
$[24,12,8]$ & $[[23,1,7]]$~\cite{steane_simple_1996}  & $[[95,1,7]]$~\cite{sullivan_code_2024}\tabularnewline
\midrule 
$[46,23,10]$~\cite{harada_extremal_2000} & $[[45,1,9]]$  & $[[185,1,9]]$ \tabularnewline
\midrule 
$[48,24,12]$ & $[[47,1,11]]$~\cite{cross_comparative_2009}  & $[[279,1,11]]$ \tabularnewline
\midrule 
$[70,35,14]$~\cite{gulliver_existence_1998} & $[[69,1,13]]$  & $[[417,1,13]]$ \tabularnewline
\midrule 
{$[80,40,16]$} & {$[[79,1,15]]$} & $[[575,1,15]]$\tabularnewline \midrule
{$[102,51,18]$}~\cite{g_philippe_tables_2007}& {$[[101,1,17]]$} & $[[777,1,17]]$\tabularnewline \midrule
 {$[104,52,20]$}& {$[[103,1,19]]$} & $[[983,1,19]]$\tabularnewline
\cmidrule{1-3}
\multirow{2}{*}{$[168,84,24]$} & \multirow{2}{*}{$[[167,1,23]]$} & $[[1317,1,21]]$\tabularnewline
 &  & $[[1651,1,23]]$\tabularnewline
\cmidrule{1-3}
\multirow{2}{*}{$[192,96,28]$} & \multirow{2}{*}{$[[191,1,27]]$} & $[[2033,1,25]]$\tabularnewline
 &  & $[[2415,1,27]]$\tabularnewline
\cmidrule{1-3}
\multirow{2}{*}{$[200,100,32]$} & \multirow{2}{*}{$[[199,1,31]]$} & $[[2813,1,29]]$\tabularnewline
 &  & $[[3211,1,31]]$\tabularnewline\bottomrule
 \\
 \multicolumn{3}{p{251pt}}{Using the best known self-dual classical codes (column 1), we construct self-dual CSS codes (column 2) with strongly transversal logical \X-gates. Doubling these $[[n,1,d]]$ self-dual CSS codes with their corresponding $[[n,1,d-2]]$ triorthogonal codes results in the shortest known $[[n,1,d]]$ triorthogonal codes for their distance (column 3). These triorthogonal codes admit the logical \T-gate via single-qubit \T-gates applied on each physical qubit, followed by some $\texttt{S}$ and $\texttt{CZ}$ gates. The stabilizer generator weights of the self dual CSS and triorthogonal codes scale roughly as $O(n)$ and $O(\sqrt{n})$, respectively.}

\end{tabular}
\end{table}

\begin{table}[t]
\caption{quadratic-residue based weak triply even codes}
\label{tab:triply_even_from_QR}
\centering
\begin{tabular}{ccc}
\toprule 
extended QR  & doubly even  & triply even*\tabularnewline
\midrule 
$[8,4,4]$ & $[[7,1,3]]$~\cite{steane_multiple_1996}  & $[[15,1,3]]$~\cite{bravyi_doubled_2015} \tabularnewline
\midrule 
  & $[[17,1,5]]$~\cite{bombin_topological_2006}  & $[[49,1,5]]$~\cite{bravyi_doubled_2015}\tabularnewline
\midrule 
$[24,12,8]$ & $[[23,1,7]]$~\cite{steane_simple_1996}  & $[[95,1,7]]$~\cite{sullivan_code_2024}\tabularnewline
\midrule 
\multirow{2}{*}{$[48,24,12]$} & \multirow{2}{*}{$[[47,1,11]]$~\cite{cross_comparative_2009}} & $[[189,1,9]]$\tabularnewline
 &  & $[[283,1,11]]$\tabularnewline
\cmidrule{1-3}
\multirow{2}{*}{$[80,40,16]$} & \multirow{2}{*}{$[[79,1,15]]$} & $[[441,1,13]]$\tabularnewline
 &  & $[[599,1,15]]$\tabularnewline
\cmidrule{1-3}
\multirow{2}{*}{$[104,52,20]$}& \multirow{2}{*}{$[[103,1,19]]$} & $[[805,1,17]]$\tabularnewline
 &  & $[[1011,1,19]]$\tabularnewline
\cmidrule{1-3}
\multirow{2}{*}{$[168,84,24]$} & \multirow{2}{*}{$[[167,1,23]]$} & $[[1345,1,21]]$\tabularnewline
 &  & $[[1679,1,23]]$\tabularnewline
\cmidrule{1-3}
\multirow{2}{*}{$[192,96,28]$} & \multirow{2}{*}{$[[191,1,27]]$} & $[[2061,1,25]]$\tabularnewline
 &  & $[[2443,1,27]]$\tabularnewline
\cmidrule{1-3}
\multirow{2}{*}{$[200,100,32]$} & \multirow{2}{*}{$[[199,1,31]]$} & $[[2841,1,29]]$\tabularnewline
 &  & $[[3239,1,31]]$\tabularnewline
\bottomrule
\\
\multicolumn{3}{p{251pt}}{Using the fact that the $[[7,1,3]]$ Steane and $[[23,1,7]]$ quantum Golay codes stem from classical quadratic-residue (QR) codes, we use other QR codes (column 1) to identify longer doubly even CSS codes (column 2), each of which admits transversal logical Clifford gates.
We then obtain the best known $[[n,1,d]]$ weak triply even codes (column 3) for distances \(9\leq d\leq 31\), each of which admits the logical \T-gate with a  transversal action of the physical \texttt{T}-gates, without any Clifford corrections. 
Each weak triply even code is obtained from its corresponding doubly even code and the previous weak triply even code via doubling~\cite{betsumiya_triply_2012,bravyi_doubled_2015,haah_towers_2018}. The stabilizer generator weights of the doubly even and the triply even* codes scale roughly as $O(n)$ and $O(\sqrt{n})$, respectively.} 
\end{tabular}

\end{table}

Quantum error correction (QEC) is an integral part of realizing reliable quantum computation. Any viable QEC solution needs to admit a fault-tolerant set of universal gates in order to suppress the noise buildup over large computations. This is usually achieved by choosing codes with fault-tolerant implementations of a non-universal ``easy'' set of gates in the Clifford group~\cite{bombin_topological_2006,gottesman_quantum_2016,katzgraber_topological_2010,yoder_surface_2017,koutsioumpas_smallest_2022,vasmer_morphing_2022}, together with a ``hard'' gate (necessary for universality) implemented by magic-state distillation~\cite{bravyi_universal_2005,campbell_roads_2017}. 

A promising alternative to resource-intensive magic-state distillation is to design codes that can implement a logical non-Clifford gate in a naturally fault-tolerant fashion.
Such codes can then be used in conjunction with Clifford-group codes to achieve universal computation via various protocols such as gauge fixing or code-switching~\cite{gong_computation_2024, heusen_efficient_2024,paetznick_universal_2013, bombin_gauge_2015, yoder_universal_2016, sullivan_code_2024,anderson_fault-tolerant_2014, jochym-oconnor_fault-tolerant_2019}.

A natural path to such fault tolerant implementations of the desired non-Clifford gate is via transversality, i.e., via a tensor product of unitary operations acting on each physical qubit.
That way, any errors occurring during such a gate cannot spread to neighboring qubits.

A popular choice of non-Clifford gate is the $\texttt{T}$-gate, $\texttt{T}=|0\rangle\langle 0|+e^{i \pi/4}|1\rangle\langle 1|$. 
Quantum codes preserved under the transversal action of $\texttt{T}$ (i.e., an action of a power of the $\texttt{T}$-gate on each physical qubit) hold much promise for universal fault-tolerant computation.
Non-stabilizer codes can admit this property at very short length (i.e., number of qubits) \(n\)~\cite{kubischta_not-so-secret_2024}, but correcting errors for such codes is not straightforward.
On the other hand, \eczoo[Qubit stabilizer codes]{qubit_stabilizer}~\cite{gottesman_stabilizer_1997,calderbank_quantum_1997} come with established error-correction protocols, and several families of \T-gate-supporting stabilizer codes are under active development. We use the standard notation $[[n,k,d]]$ to denote a stabilizer code with $n$ physical qubits, $k$ encoded qubits and distance $d$.

Most \T-gate-supporting codes have low distance given a number of qubits \(n\) relative to codes which do not admit transversal non-Clifford gates~\cite{grassl_markus_bounds_2007}.
While many of these codes benefit from having geometrically local structure, recent advances in ion-trap~\cite{monroe_scaling_2013}, photon~\cite{rudolph_why_2017} and neutral-atom~\cite{bluvstein_logical_2024} architectures allow one to relax the locality requirement and consider codes that have non-local stabilizers, but that require a lower number of physical qubits \(n\) to realize the same distance. 

Motivated by finding such codes, we allow ourselves to look at codes with potentially non-local and high-degree stabilizers to optimize the physical qubit overhead required to realize any given distance. We focus on codes encoding a single logical qubit (i.e. $k=1$) and construct three types of qubit stabilizer codes, first of which realize the logical Clifford group transversally and the latter two admit transversal implementations of the logical \T-gate. These codes, to the best of our knowledge, are the shortest to realize their respective transversal gate sets for their corresponding distances. We present below a brief overview of these families and our main results.
\begin{enumerate}
    \item \textit{Transversal Clifford gates}: Leveraging the well-established class of (classical) quadratic-residue (QR) codes, we identify a family of \([[n,1,d]]\) QR CSS codes of distance up to \(31\) that are all doubly even (see Section~\ref{subsec:divisible-codes} for definition). These are the shortest qubit stabilizer codes with their respective distances to realize the full Clifford group transversally (see Table~\ref{tab:triply_even_from_QR} for complete list).
    We outline how to extend these codes to an infinite family with growing distance.

    \item \textit{Transversal \T-gate}: Plugging these codes into a doubling procedure ~\cite{betsumiya_triply_2012,bravyi_doubled_2015,haah_towers_2018}, we construct a family of weak triply even codes (see Section~\ref{subsec:divisible-codes} for definition) which showcase growing distance with the lowest qubit overhead of any known examples.
    To our knowledge, these are the shortest qubit stabilizer codes with such distance to realize a \(\texttt T\)-gate transversally, without any Clifford corrections (see Table~\ref{tab:triply_even_from_QR} for complete list).
    
    \item \textit{Transversal \T-gate (up to Clifford corrections)}:
    Mapping the best known classical self-dual codes to self-dual CSS codes and then doubling them, we obtain a new family of triorthogonal codes (see Section~\ref{subsec:triortho} for definition). This family represents, to the best of our knowledge, the shortest triorthogonal codes~\cite{bravyi_magic-state_2012,haah_towers_2018,nezami_classification_2022} for their given distances (see Table~\ref{tab:triortho} for complete list). These codes implement the \T-gate on the encoded qubit via a transversal implementation, followed by some Clifford gates on the physical qubits.

\end{enumerate}

Computational limitations in computing high code distances prevent us from making exact statements about the parameters of the larger members of our code families. Nevertheless, we are able to prove the existence of infinite families of all three kinds of the aforementioned codes for increasing distance, provide bounds on their parameters, and describe how to obtain them.

Leaving the details for Section~\ref{sec:overview}, we introduce here some relevant code properties and their implications. Doubly even (DE) and weak doubly even (DE*) codes with single logical qubits allow certain transversal implementations of the Clifford gates to preserve the codespace. Analogously, triply even (TE) and weak triply even (TE*) codes with single logical qubits allow certain transversal implementations of the non-Clifford \T-gate to preserve the codespace.

The doubling mapping takes a triorthogonal code and a self-dual code to yield a triorthogonal code with potentially higher distance. 
In terms of quantum code notation, this mapping takes a self-dual $[[n_1,1,d_1]]$ and a triorthogonal $[[n_2,1,d_2]]$ code to a triorthogonal $[[2n_1+n_2,1,\min (d_1,d_2+2)]]$ code. In certain special cases, such as the doubled color codes~\cite{bravyi_doubled_2015}, the same mapping can take a DE* code (instead of a self-dual code) and a TE* code (instead of a triorthogonal code) to output a TE* code,   effectively ``doubling'' the divisibility of the input DE* code. 

The family of TE* doubled color codes was previously constructed using the DE \eczoohref[color codes]{2d_color}~\cite{bombin_exact_2007,bombin_topological_2006} by repeatedly applying the doubling mapping to the shortest distance-three TE code~\cite{bravyi_doubled_2015}. Doubling has also been used to construct a $[[95,1,7]]$ TE code using the $[[49,1,5]]$ TE and the $[[23,1,7]]$ DE \eczoohref[quantum Golay code]{qubit_golay}~\cite{sullivan_code_2024}. 
The 49-qubit code is the shortest possible distance-five TE code~\cite{bravyi_magic-state_2012}, while the 95-qubit code currently holds the record for the shortest distance-seven code.

Our key observation is that the \eczoohref[classical Golay code]{golay}~\cite{marcel_je_golay_notes_1949} is a DE \eczoohref[quadratic residue (QR) code]{binary_quad_residue}~\cite{macwilliams_theory_1977}.
We use the classical QR family to identify DE quantum QR codes via the \eczoohref[CSS construction]{qubit_css}~\cite{calderbank_good_1996,steane_error_1996,steane_multiple_1996}, each of which realizes the single-qubit Clifford group transversally. Some of these codes have been noted before as part of the quantum QR family~\cite{lai_construction_2011}, and some codes with the same parameters exist in the Bose–Chaudhuri–Hocquenghem (BCH) family~\cite{grassl_quantum_1999}.
Their fault-tolerant gates have not been studied however, to our knowledge, with the exception of Ref.~\cite{cross_comparative_2009} for the case of \(n\leq79\). Doubling these codes then yields the new TE* family. 

We also identify some of the smallest self-dual classical codes which can be mapped to self-dual CSS codes. We use the same doubling mapping to obtain the best known examples of triorthogonal codes, which are shorter than the constructed TE* family. This improvement is achieved at the cost of extra $\texttt{S}$ and $\texttt{CZ}$ gates required by the triorthogonal family to achieve the same logical action~\cite{bravyi_magic-state_2012}. Hence, the TE* family is potentially more resource efficient as it admits truly transversal implementations of the logical $\texttt{T}$-gate. 

A recent work~\cite{shi_triorthogonal_2024} presents an alternative algorithm to construct triorthogonal codes from self-dual classical codes, presenting codes with distance $d \leq 3$. 
Complementing this work, we focus on self-dual codes that are also DE and use the combination of the CSS construction and the doubling map, which turns out to boost the code distance substantially. Since QR codes are cyclic, we also solve the open problem posed in Ref.~\cite{shi_triorthogonal_2024} of constructing an infinite family of triorthogonal codes using cyclic codes.

We begin by presenting an overview of the different code families relevant to transversality of \T-gates, describing their properties, how they relate to each other and clarifying some common ambiguities related to them in Section~\ref{sec:overview}. We then describe the doubling map and focus on using it to generate weak triply even codes in Section~\ref{sec:doubling}. Section~\ref{sec:qr_codes} constructs doubly even and weak triply even codes based on the quadratic-residue code family (Table~\ref{tab:triply_even_from_QR}). Similar methodology is also used to  construct the best known examples of triorthogonal codes (Table~\ref{tab:triortho}) using self-dual classical codes. We conclude in Section~\ref{sec:conclusion}.

\section{Classes of \T-gate codes}\label{sec:overview}
Various classes of CSS codes have been defined which can allow a \textit{transversal \T-gate} (i.e. $\texttt{T}^j$ acting on each qubit for some $j\in \mathbb{Z}$) or \textit{strongly transversal \T-gate} (all \(j=1\)) on the physical qubits to implement either logical identity or logical $\texttt{T}$-gate(s) on it, possibly with some Clifford corrections. We review those classes and their related concepts here with a focus on describing how they relate to each other (see Figure~\ref{fig:code_families} and Table~\ref{tab:classes_summary}). 

An $[[n,k,d]]$ CSS code is denoted as $CSS(X,C_2;Z,C_1^\perp)$ where $C_2$, $C_1$ are the classical codes that generate the groups of \X-stabilizers, \X-logical operators, respectively, and $C_1^\perp$ is the set of all binary strings orthogonal to \(C_{1}\) under the standard binary inner product~\cite{macwilliams_theory_1977}. We separate the non-trivial logical gate generators $C_{1,(1)}$ from the stabilizers $C_{2}$ as 
\begin{equation}
    C_1 = \left[\begin{array}{c}
   C_{1,(1)}  \\
   \midrule
   C_{2}\\
\end{array}\right] .
\end{equation}
In order to exclude codes with $k=0$, we impose the condition that $C_2 \subsetneq C_1$ implying that $C_{1,(1)}$ is non-empty. With every CSS code, there is an associated character vector $y\in\mathbb{F}_2^{n}$ which determines the sign of every \Z-stabilizer \Z$(z)$ via the equation $\epsilon_z = (-1)^{zy^T}$ for $z\in C_1^\perp$~\cite{hu_designing_2022}.

\begin{table*}[t]
  \centering

    \caption{A summary of the \T-gate code families }
    \label{tab:classes_summary}
\begin{tabular}{l | c c c c c c }
    \toprule 
\makecell[c]{\textbf{transversal}\\  \textbf{gate}}  & \textbf{CSS-T}    & \textbf{TE}    & \textbf{TE*}    
                & \textbf{triorthogonal}    & \makecell[c]{\textbf{generalized}\\  \textbf{triorthogonal}}    &    $\mathbf{y=\bar{0}}$           \\
    \midrule 
${\T}^{\otimes n}$ preserves codespace   & \checkmark  &  {}  & {}  & {}  & {}  & {} \\
${\T}^{\otimes n}\sim \bar{\texttt{I}}^{\otimes k}$    & \checkmark  &  {}  & {}  & $\times$  & {}  & {} \\
${\T}^{\otimes n}$ preserves codespace   & \checkmark  &  $\medbullet$  & $\medbullet$  & {}  & {}  & \checkmark \\
${\T}^{\otimes n}\sim\bar{\texttt{T}^m}^{\otimes k}$ with $m$ odd    & \checkmark  &  $\medbullet$  & $\medbullet$  & \checkmark  & {}  & $\medbullet$ \\
${\T}^{\otimes n}\sim\bar{\texttt{T}^m}^{\otimes k}$ (known for $k=1$)  & $\medbullet$  &  \checkmark  & $\medbullet$  & {}  & {}  & \checkmark \\
${\T}_{all}\sim\bar{\texttt{T}^m}^{\otimes k}$ (known for $k=1$)   & $\medbullet$  &  {}  & \checkmark  & {}  & {}  & \checkmark \\
${\T}^{\otimes n}\sim\bar{\texttt{T}}^{\otimes k}$ up to Clifford corrections   & {} & {}  & {}  & \checkmark  & $\medbullet$  & \checkmark \\
${\T}^{\otimes n}$ induces logical \T, \CS, \CCZ$\textrm{ }$up to Clifford corrections   & {} & {}  & {}  & $\medbullet$  & \checkmark  & \checkmark \\
    \bottomrule  
\end{tabular} 
\begin{tabular}{p{502pt}} 
\\
Column 1 describes the relevant transversal gate and its logical action on the $[[n,k,d]]$ CSS code satisfying the checkmarked (\checkmark) code properties. The solid dots ($\medbullet$) indicate code properties implied by the checkmarked code families admitting corresponding transversal action of the first column. The cross ($\times$) indicates code properties inconsistent with the checkmarked code families and transversal action of the same row. $A\sim\bar{B}$ denotes the action of $B$ on the encoded space induced by the action of $A$ on the physical qubits, $\T_{all}$ is the partitioned transversal \T-gate defined in Lemma~\ref{lemma:TE*gate} and $y=\bar{0}$ denotes all the \Z-type stabilizers being positively signed.
\end{tabular}

\end{table*}

\begin{figure}[t]
  \centering
  \includegraphics[width=0.5\textwidth]{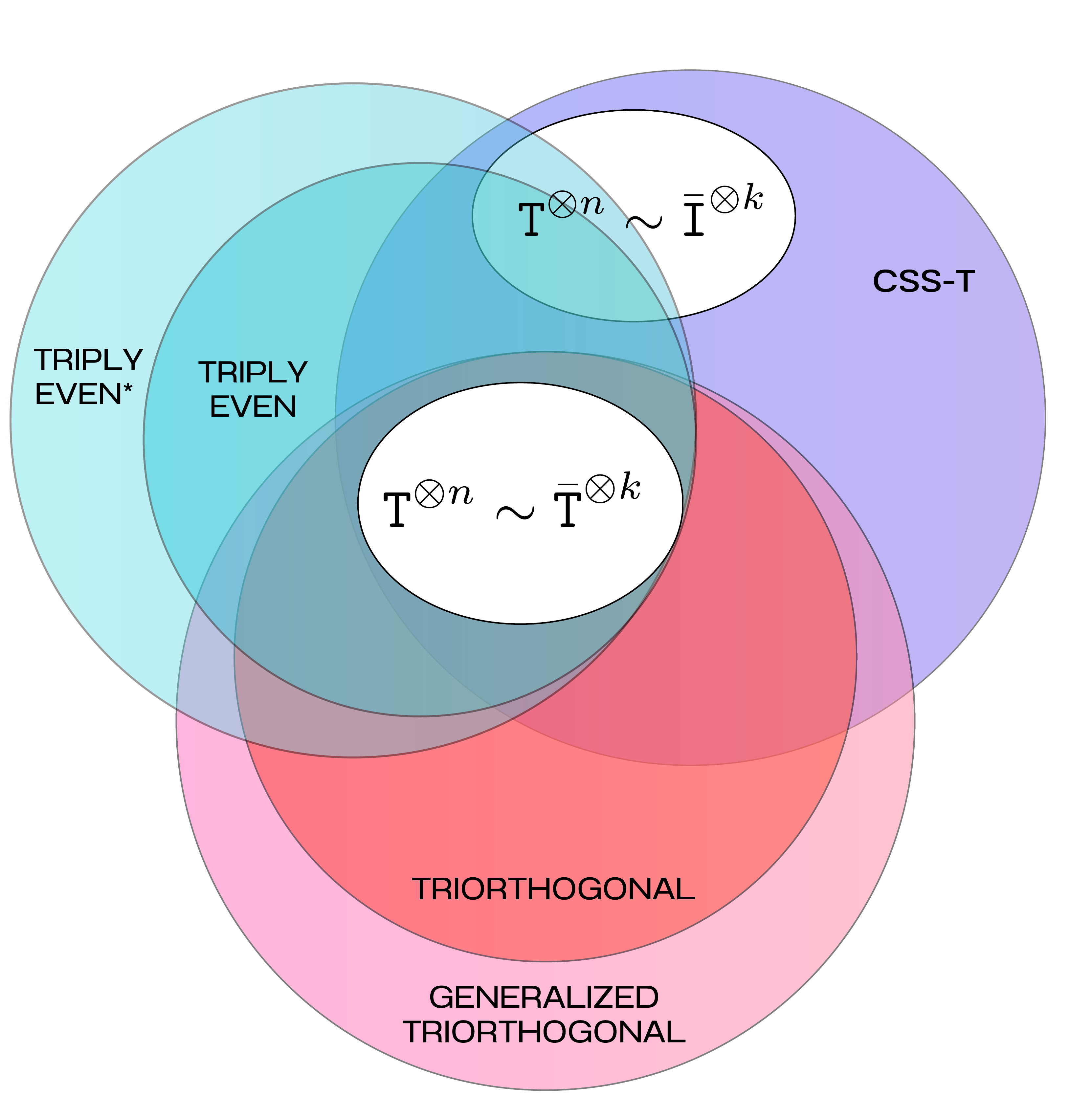}
  \caption{Relationship between various $\texttt{T}$-gate code families. The white bubbles represent codes where the strongly transversal \T-gate implements the logical \T-gate ($\bar{\T}$) or the logical identity ($\bar{\texttt{I}}$) on each encoded qubit. 
  }
  \label{fig:code_families}
\end{figure}

\subsection{CSS-T codes}\label{subsec:css-t}
First introduced by Rengaswamy et al.~\cite{rengaswamy_optimality_2020,rengaswamy_classical_2020}, \textit{\eczoohref[CSS-T codes]{css-t}} are the most general family of CSS codes for which a transversal implementation of the \T-gate on the physical qubits preserves the codespace by performing e.g., a logical $\texttt{T}$-gate~\cite{campbell_roads_2017,rengaswamy_optimality_2020,koutsioumpas_smallest_2022,sullivan_code_2024,haah_towers_2018,bravyi_magic-state_2012,vasmer_morphing_2022, rengaswamy_classical_2020,tansuwannont_achieving_2022} or the logical identity~\cite{rengaswamy_optimality_2020,andrade_css-t_2023,camps-moreno_algebraic_2024,berardini_structure_2024,rengaswamy_classical_2020} on the encoded information.

We define the CSS-T family to be most general class of CSS codes for which a strongly transversal implementation of the \T-gate on the physical qubits preserves the codespace. This modified definition is adopted to be consistent with the many recent works~\cite{camps-moreno_algebraic_2024,andrade_css-t_2023,berardini_structure_2024,berardini_asymptotically_2024,camps-moreno_toward_2024,camps-moreno_binary_2024}. A theorem by Hu et al.~\cite{hu_designing_2022} provides a simple description of the CSS-T class:
\begin{theorem}\label{thm:hu_css-t}
   A strongly transversal \T-gate (i.e. \T-gate acting on every physical qubit) preserves the codespace $CSS(X,C_2;Z,C_1^\perp)$ with the character vector $y$ if and only if
    \begin{equation}
        |x|-2|x \star z| = 0\mod{8} \quad \forall\,\,x\in C_2,z\in C_1+y,
    \end{equation}
    where $|x|$ is the Hamming weight of $x$ and $x\star z$ is the element-wise product of $x$ and $z$.
\end{theorem}
The following corollary further simplifies this description for the most widely considered case of all positively signed \Z-stabilizers.
\begin{corollary}\label{cor:hu_css-t}
    The strongly transversal \T-gate preserves the codespace $CSS(X,C_2;Z,C_1^\perp)$ with all positively signed \Z-stabilizers (i.e., character vector $y=0$) if and only if
    \begin{eqnarray}
        |x| &=& 0 \mod{8} \quad \forall\,\, x\in C_2, \\
        |x\star z| &=& 0 \mod{4} \quad \forall\,\, x\in C_2, z\in C_1.
    \end{eqnarray}
\end{corollary}
The CSS-T codes can be understood as a composition of two components: a) a \textit{CSS-T pair} of classical codes $(C_1,C_2)$, and b) the character vector determining the signs of the \Z-stabilizers. Camps-Moreno et al.~\cite{camps-moreno_algebraic_2024} study the algebraic structure of the required classical codes and provide simplified conditions for $(C_1,C_2)$ to be a CSS-T pair of classical codes, paraphrased as the following theorem.
\begin{theorem}\label{thm:css-t_conditions}
    Let $C_1$ and $C_2$ be two classical codes. Then, the following statements are equivalent:
    \begin{enumerate}
        \item $(C_1, C_2)$ form a CSS-T pair.
        \item $C_2 \subset C_1$, $C_2$ is even-weighted, and for any $x \in C_2$, the code $C_1^{supp(x)}$ is self-orthogonal.
        \item $C_2\subset C_1 \cap (C_1^{\star 2})^\perp$.
        \item $C_1^\perp + C_1^{\star 2} \subset C_2^\perp$
    \end{enumerate}
    where $C_1^{\star 2} = \{x \star y \, | \, x,y\in C_1\}$,  $C_1^{supp(x)} = \{c\in C_1 \, | \, supp(c)\subseteq supp(x) \}$ and $supp(x) = \{i \,|\, x_i \neq 0\}$.
\end{theorem}
Intuitively, the set $C_1^{supp(x)}$ consists of all the vectors in $C_1$ which are non-zero only in the positions where $x$ is non-zero.

\begin{remark}
    The CSS-T pair conditions of Theorem~\ref{thm:css-t_conditions} are necessary but not sufficient to define a CSS-T code.
\end{remark}

We emphasize that Theorem~\ref{thm:css-t_conditions} only deals with one component of defining a CSS-T code, namely, the CSS-T pair of classical codes. This component is necessary but also requires a compatible character vector in order to define a CSS-T code. We found that all the examples of CSS-T pairs we looked at admit some character vector which satisfies the conditions of Theorem~\ref{thm:hu_css-t} but whether such a character vector exists for every CSS-T pair i.e. whether every CSS-T pair of classical codes can be used to define a CSS-T quantum code is still an open question. Thus, the conditions of Camps-Moreno et al. (summarized in Theorem~\ref{thm:css-t_conditions}) should be viewed as a set of necessary, but not sufficient conditions for the underlying classical codes to define a CSS-T code.

\subsection{Triorthogonal codes}\label{subsec:triortho}
\eczoohref[\textit{Triorthogonal codes}]{quantum_triorthogonal}, defined by Bravyi and Haah~\cite{bravyi_magic-state_2012}, are $CSS(X,C_2;Z,C_1^\perp)$ codes for which the generating matrix of the code $C_1$ is triorthogonal according to the definition:
\begin{definition}\label{def:triortho}
    A binary matrix $M$ of dimension $m \times n$ is called triorthogonal if and only if 
    \begin{eqnarray}
       |M_a\star M_b| = 0 \mod{2} \quad \forall \,\, 1\leq a < b \leq m ,\\
       |M_a\star M_b \star M_c| = 0 \mod{2} \quad \forall \,\, 1\leq a < b <c \leq m .
    \end{eqnarray}
\end{definition}

 Triorthogonal codes also demand that the rows of $C_{1,(1)}$ and $C_{2}$ be odd and even weighted, respectively. Such a CSS code admits the logical \T-gate via a strongly transversal implementation, supplemented by additional \texttt{S} and \texttt{CZ} Clifford gates. We note that the original definition of triorthogonal codes include the case where $C_2=C_1$, while we assume $C_2\subsetneq C_1$. In the former case, the logical action of any physical gate could at most be the identity, so the analysis extends trivially.

\begin{remark}
    Not all triorthogonal codes are CSS-T codes.
\end{remark}

Triorthogonal codes are defined with each \Z-stabilizer positively signed, and must satisfy the conditions of Corollary~\ref{cor:hu_css-t} in order to be CSS-T. In particular, a triorthogonal code must also be triply even and have a doubly even overlap (i.e. weight of the overlap divisible by $4$) between all \X-type stabilizers and logical operators for the strongly transversal \T-gate to preserve the codespace. For triorthogonal codes which do not satisfy these conditions, the strongly transversal \T-gate takes the logical states out of the codespace, and are brought back (along with the logical \T-action) only after the application of the additional Clifford gates as prescribed in Ref.~\cite{bravyi_magic-state_2012}. 

We demonstrate this with the example of the $[[3,1,1]]$ $CSS(X,C_2;Z,C_1^\perp)$ code defined by the triorthogonal matrix
\begin{equation}
    C_1 = \left[\begin{array}{ccc}
   1 & 1 & 1  \\
   \midrule
   1 & 1 & 0\\
\end{array}\right] .
\end{equation} 
The logical code states can be written as
\begin{eqnarray}
    |\bar{0}\rangle &=& \frac{1}{\sqrt{2}}\left(|000\rangle+|110\rangle\right), \\
    |\bar{1}\rangle &=& \frac{1}{\sqrt{2}}\left(|111\rangle+|001\rangle\right).
\end{eqnarray}
We can explicitly check that the action of the strongly transversal \T-gate takes these states out of the codespace:
\begin{eqnarray}
    \texttt{T}^{\otimes3}|\bar{0}\rangle &=& \frac{1}{\sqrt{2}}\left(|000\rangle+i|110\rangle\right),\label{eqn:T30example} \\
    \texttt{T}^{\otimes3}|\bar{1}\rangle &=& \frac{1}{\sqrt{2}}e^{i\pi/4}\left(i|111\rangle+|001\rangle\right). \label{eqn:T31example}
\end{eqnarray}
It is clear that this codespace is not preserved under the action of the strongly transversal \T-gate and hence, the code is not CSS-T. However, an additional application of $\texttt{S}^\dagger$ on the second qubit gets rid of the additional $i$ factors in Eqns.~\ref{eqn:T30example} and \ref{eqn:T31example} and leads to the implementation of the logical \T-gate on the codespace. This demonstrates the need of Clifford corrections for the implementation of the logical \T-gate in triorthogonal codes. A more non-trivial example is the $[[185,1,9]]$ triorthogonal code in Table~\ref{tab:triortho} which admits the weight $140$ operator $\texttt{I}^{\otimes 45}\texttt{X}^{\otimes 140}$ as an \X-stabilizer. Since $8$-divisibility is necessary for a triorthogonal code to be CSS-T and $140 \neq 0\mod{8}$, we conclude that this code is neither triply even nor CSS-T.

\begin{remark}
    A strongly transversal \T-gate cannot induce the logical identity on the codespace for any triorthogonal code.
\end{remark}
A previous result from Rengaswamy et al.~\cite{rengaswamy_optimality_2020} states that for a $CSS(X,C_2;Z,C_1^\perp)$ code, where the strongly transversal \T-gate acts as the logical identity on the codespace, the matrix $C_1$ must be triorthogonal and all \X-logical operators be even-weighted. On the contrary, triorthogonal codes are defined with odd-weighted \X-type logical operators. Hence, these two classes of codes do not overlap (as shown in Figure~\ref{fig:code_families}). 

The same can also be inferred by analyzing the transversal action of the \T-gate on triorthogonal codes. A triorthogonal code always admits the logical \T-gate upon the action of the strongly transversal \T-gate, supplemented by some Clifford corrections. However, if it were to admit the logical identity upon the action of the strongly transversal \T-gate, then no Clifford gate could induce the logical \T-gate on the codespace. Hence, a triorthogonal code cannot admit the logical identity induced by the action of the strongly transversal \T-gate.

We also note here the existence of a superclass of triorthogonal codes called \textit{generalized triorthogonal} codes. These generalized codes implement the logical \T, \CS, and \CCZ-gates on different logical qubits via a strongly transversal implementation of the \T-gate, supplemented by Clifford corrections. Figure~\ref{fig:code_families} and Table~\ref{tab:classes_summary} describe how this class relates to the other classes. We refer the interested reader to the paper by Haah and Hastings~\cite{haah_codes_2018} for more details.

\subsection{Divisible codes}\label{subsec:divisible-codes}
\eczoohref[Divisible codes]{divisible} have long been studied in classical coding theory~\cite{macwilliams_theory_1977,  ward_divisible_1981, betsumiya_triply_2012,kurz_divisible_2023}.
They are related to CSS codes whose codespace is preserved under transversal phase gates and are defined as follows:
\begin{definition}
    A codespace $C$ is called $\Delta$-divisible if
    \begin{equation}
        |x|=0\mod{\Delta} \quad \forall \,\,x\in C.
    \end{equation}
\end{definition}
A $\Delta$-divisible CSS code is analogously defined by the Pauli weight of each \X-stabilizer being divisible by $\Delta$.
\begin{definition}
    A $CSS(X,C_2;Z,C_1^\perp)$ code is called $\Delta$-divisible if $C_2$ is a $\Delta$-divisible classical code.
\end{definition}
$4$-divisible and $8$-divisible codes are also termed doubly (DE) and triply even (TE) codes, respectively. 
Divisibility of CSS codes is closely related to the logical action induced by transversal phase gates~\cite{hu_divisible_2022,landahl_complex_2013,campbell_unified_2017,haah_codes_2018}. For example: $8$-divisibility is a necessary condition for a CSS code to admit strongly transversal logical \T-gates~\cite{rengaswamy_optimality_2020}. 

To understand the impact of the \X-stabilizer weights on the action of transversal phase gates, it is instructive to look at the example of the DE Steane $[[7,1,3]]$ code with $y=0$. Its $C_1$ matrix and codewords are described as
\begin{eqnarray}
    C_1 &=& \left[\begin{array}{ccccccc}
   1 & 1 & 1 & 1 & 1 & 1 & 1 \\
   \midrule 
   0 & 0 & 0 & 1 & 1 & 1 & 1  \\
   0 & 1 & 1 & 0 & 0 & 1 & 1\\
    1 & 0 & 1 & 0 & 1 & 0 & 1\\
\end{array}\right], \\
|\bar{0}\rangle &=& \sum_{v\in \textrm{span}(C_2)} |v\rangle, \\
    |\bar{1}\rangle &=& \sum_{v\in \textrm{span}(C_2)} |v+\bar{1}\rangle .
\end{eqnarray} 

The qubit configurations in the $|\bar{0}\rangle$ state are the codewords of the linear code $\textrm{span}(C_2)$ and those of the $|\bar{1}\rangle$ state are the same codewords added to the binary representation of the \X-logical operator $\bar{1}$. The action of the strongly transversal ${\texttt{S}}$-gate (where ${\texttt{S}}=|0\rangle\langle 0| + i|1\rangle\langle 1|$) is then described by 
\begin{eqnarray}
    \texttt{S}^{\otimes 7}|\bar{0}\rangle &=& \sum_{v\in \textrm{span}(C_2)} \exp\left( 
 i\frac{\pi}{4}|v|\right)|v\rangle, \\
   \texttt{S}^{\otimes 7} |\bar{1}\rangle &=& \sum_{v\in \textrm{span}(C_2)} \exp\left(
 i\frac{\pi}{4}|v+\bar{1}|\right)|v+\bar{1}\rangle
 .
\end{eqnarray}

Since the Steane code is a DE code,  $|v|=0\mod{4}\,\forall\,v\in \textrm{span}(C_2)$, resulting in a trivial action on the $|\bar{0}\rangle$ state. It is also not hard to check that $|v+\bar{1}|=3\mod{4}\,\forall\,v\in \textrm{span}(C_2)$, resulting in an overall phase of $-i$ on the $|\bar{1}\rangle$ state. Consolidating these two phases, we conclude that the strongly transversal $\texttt{S}$ gate on the Steane code results in the logical $\bar{\texttt{S}^\dagger}$ gate. Equivalently, the strongly transversal $\texttt{S}^\dagger$ gate translates to the logical $\bar{\texttt{S}}$ gate. 

In an analogous manner, TE codes have a direct impact on the action of the transversal \T-gates. More generally, the weights of the \X-stabilizers directly impact the transversal action of the phase gates for codes with trivial character vector $y$. For codes with a non-trivial character vector, the weights have more indirect but crucial implications for the transversal phase gates, as we saw in Theorem~\ref{thm:hu_css-t}.

Defined by Bravyi and Cross~\cite{bravyi_doubled_2015}, we refer to the following related class of codes as \textit{weakly divisible codes}:
\begin{definition}\label{def:gen_delta_div}
    A codespace $C \subseteq \mathbb{F}^n_2$ is weakly $\Delta$-divisible if there exists disjoint subsets $M^\pm\subseteq [n]$ such that 
    \begin{equation}
        |c\cap M^+|-|c\cap M^-|=0 \mod{\Delta} \quad \forall\,\, c\in C 
        \end{equation}
\end{definition}
Analogously, we define the quantum analogue of weakly divisible codes:
\begin{definition}\label{def:quantum_gen_delta_div}
   A $CSS(X,C_2;Z,C_1^\perp)$ code is called weakly $\Delta$-divisible if $C_2$ is weakly $\Delta$-divisible.
\end{definition}

We will denote weakly $4$-divisible and $8$-divisible codes as DE* and TE* codes, respectively.  Note that a weakly $\Delta$-divisible code with $M^+=[n], M^- = \emptyset$ is a $\Delta$-divisible code according to the standard notion of divisibility.

This generalization is especially useful for codes encoding a single logical qubit i.e., $[[n,1,d]]$ codes, which is the primary focus of our discussion. TE*(DE*) $[[n,1,d]]$ codespaces are preserved under well-defined transversal \T(\Sgate)-gates, according to the following Lemmas~\cite{bravyi_doubled_2015}:
\begin{lemma}\label{lemma:TE*gate}
    A TE* $CSS(X,C_2;Z,C_1^\perp)$ code encoding a single logical qubit and admitting a strongly transversal logical \X-gate applies the gate $\bar{\T}^m$ on the logical qubit upon the action of the transversal gate 
    \begin{equation}
        \T_{all}= \displaystyle {\prod_{i\in M^+}}\T_i \displaystyle{\prod_{j\in M^-}}\T_j^{-1} 
    \end{equation}
    where $m=|M^+|-|M^-| \mod{8}$.
\end{lemma}
\begin{lemma}
    A self-dual DE* $CSS(X,C_2;Z,C_1^\perp)$ code encoding a single logical qubit applies the gate $\bar{\Sgate}^m$ on the logical qubit upon the action of the transversal gate 
    \begin{equation}
        \Sgate_{all}= \displaystyle {\prod_{i\in M^+}}\Sgate_i \displaystyle{\prod_{j\in M^-}}\Sgate_j^{-1} 
    \end{equation}
    where $m=|M^+|-|M^-| \mod{4}$.
\end{lemma}
Thus, depending on the value of $m$, a TE* (DE*) code can implement the \T(\Sgate)-gate on each logical qubit via transversal physical gates, without any Clifford corrections. 

Every self-dual code admits the strongly transversal logical \Hgate-gate, and every CSS code admits the strongly transversal logical \CNOT-gate between code blocks. Hence, a self-dual DE* code with odd $m=|M^+|-|M^-|$ implements the full logical Clifford group transversally. 
In particular, an $[[n,1,d]]$ self-dual DE code with odd $n$ ($n = 1$ modulo 4) admits a (strongly) transversal implementation of the full logical Clifford group.

DE CSS codes are easily obtainable from classical \eczoohref[self-dual codes] {self_dual}~\cite{macwilliams_theory_1977,rains_self-dual_2002,kurz_divisible_2023} (i.e. $C_{\textrm{sd}}=C_{\textrm{sd}}^\perp$) DE codes via a general mapping from self-dual classical to self-dual CSS codes. Every classical DE code is necessarily self-dual, whereas a quantum DE code need not be self-dual, depending on the choice of the \Z-stabilizers. We note that in the classical coding literature, self-dual and DE codes are also commonly termed as self-dual codes of type-I and type-II, respectively.

For every self-dual $\Delta$-divisible code $C_{\textrm{sd}}$ with parameters $[n,n/2,d]$, there exists an $[[n-1,1,\geq d-1]]$ $\Delta$-divisible $CSS(X,C^\perp;Z,C^\perp)$ code. Note that both $n$ and $\Delta$ have to be even for $C_{\textrm{sd}}$ to be self-dual. The mapping can be summarized as (also see Figure~\ref{fig:selfdual-to-css}):

\begin{enumerate}
    \item Puncture the self-dual code $C_{\textrm{sd}}$ in the last position (i.e. remove the last column of the generator matrix) to obtain $C$.
    \item Use the code's dual space $C^\perp$ to define the $\texttt{X}$ and $\texttt{Z}$ type stabilizers of the CSS code, with the stabilizer generator matrix
    $$S=\begin{pmatrix}
    S_Z & 0\\
    0 & S_X
    \end{pmatrix}=\begin{pmatrix}
    C^\perp & 0\\
    0 & C^\perp
    \end{pmatrix}.$$
\end{enumerate}

This construction is similar to the mapping from triorthogonal matrices to triorthogonal codes~\cite{haah_towers_2018, bravyi_magic-state_2012} and is an example of a general stabilizer code conversion~\cite{ketkar_nonbinary_2005}. Utilizing it for the special case of DE self-dual classical codes, it is simple to show that the result is a DE CSS code encoding one logical qubit using the following Lemma.

\begin{lemma}\label{lemma:Cperp-seldual}
Let $C_{\textrm{sd}}$ be a $[n,n/2,d]$ $\Delta$-divisible self-dual code for some $\Delta$ divisible by $2$, $C$ be the $[n-1,n/2,\geq d-1]$ code obtained by puncturing the last position of $C_\textrm{sd}$, and $C^\perp$ be the $[n-1,n/2-1,d^\perp]$ code dual to $C$. Then, $C^\perp$ is a $\Delta$-divisible and weakly self-dual classical code i.e. $C^\perp \subseteq C$. 
\end{lemma}

\begin{proof}
    Consider the code $E_0(C^\perp)$, the $[n,n/2-1,d^\perp]$ code obtained by extending $C^\perp$ with a zero position as the last column. Then, 
    \begin{eqnarray}
        &(c^\perp,0).(c,x) = 0 \,\forall\,c^\perp \in C^\perp,c \in C \textrm{ and } x\in \{0,1\}\nonumber \\
        &\implies E_0(C^\perp)\subset C_\textrm{sd}^\perp=C_\textrm{sd}
    \end{eqnarray}
    Thus, $E_0(C^\perp)$ inherits the divisibility and self-orthogonality of $C_\textrm{sd}$. These properties are preserved under the removal of the $0$ column. Thus, we conclude that $C^\perp$ is also $\Delta$-divisible and weakly self-dual. 
\end{proof}
\begin{figure}[t]
  \centering
  \includegraphics[width=0.49\textwidth]{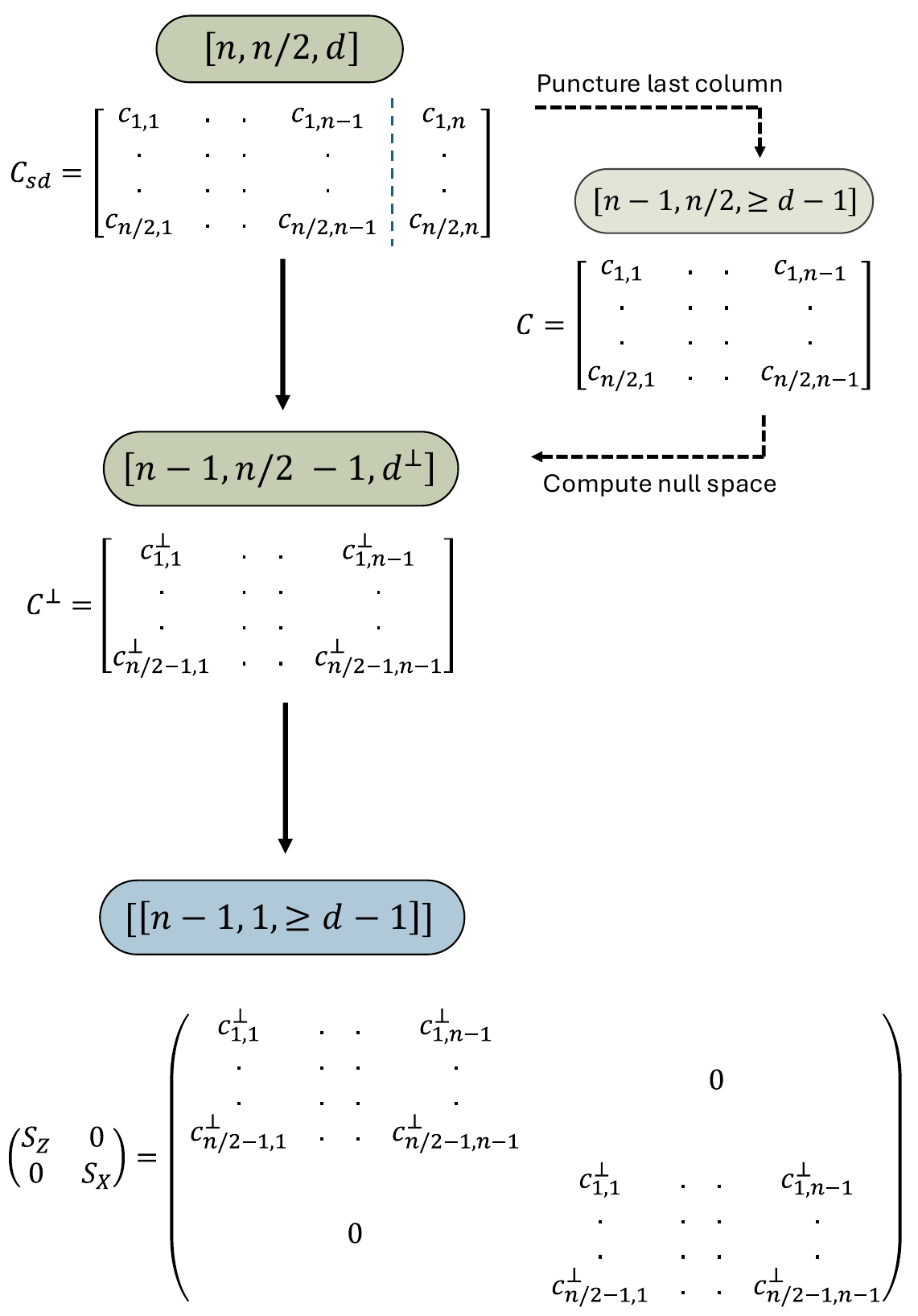}
  \caption{The self-dual code $C_{\textrm{sd}}$ is used to obtain $C^\perp$ which forms the $\texttt{X}$ and $\texttt{Z}$ stabilizers of the doubly even CSS code. The branch arising on the right details how to obtain $C^\perp$ from $C_{\textrm{sd}}$. }
  \label{fig:selfdual-to-css}
\end{figure}

The preceding lemma can be used to construct a valid quantum code using the procedure described before (Figure \ref{fig:selfdual-to-css}). Using a DE classical code $C_{\textrm{sd}}$ results in a DE weakly self-dual CSS code. Now, it is a property of every CSS code that its distance is at least as large as the minimum of the dual distances of the classical codes that define its $\texttt{X}$ and \Z-stabilizers~\cite{daniel_gottesman_surviving_2024}. So the distance of the constructed CSS code is at least as large as that of the punctured code $C$ with distance $d-1$. Hence, the constructed CSS code has distance $\geq d-1$. 

To define logical Pauli operators on the encoded qubit, consider the vector of all ones $\bar{1}=1^{\otimes n-1}$. Since $n-1$ is odd, and $C^\perp$ has all codewords of even weight, $\bar{1}$ is not a codeword in \(C^\perp\) and is orthogonal to all vectors in $C^\perp$. This implies that $\texttt{X}^{\otimes n-1}$ preserves the CSS codespace but is not a stabilizer element. Hence, it can be designated as the logical $\texttt{X}$-gate. Similarly, $\texttt{Z}^{\otimes n-1}$ is a logical operator which anticommutes with the logical $\texttt{X}$.
Hence, it can be designated as the logical $\texttt{Z}$-gate. 

\section{Doubling}\label{sec:doubling}
Doubling was first introduced by Betsumiya and Munemasa~\cite{betsumiya_triply_2012} as a map to construct triply even from doubly even classical codes. This map can be generalized~\cite{bravyi_doubled_2015,sullivan_code_2024}
to CSS codes encoding single logical qubits. 

Given a self-dual code $CSS(X,C_{2;sd};Z,C_{1;sd}^\perp)$ with parameters $[[n_{sd},1,d_{sd}]]$ (with odd $n_{sd}$) and a triorthogonal code $CSS(X,C_{2;tri};Z,C_{1;tri}^\perp)$ with parameters $[[n_{tri},1,d_{tri}]]$ which admits a strongly transversal logical \X-gate, doubling results in a triorthogonal code  $CSS(X,C_{2;dbl};Z,C_{1,dbl}^\perp)$ with parameters $[[2n_{sd}+n_{tri},1,\min(d_{sd},d_{tri}+2)]]$ where
\begin{equation}\label{eqn:dbled_logical}
\left[ \begin{array}{c}
   C_{1;dbl,(1)}  \\
   \midrule
   C_{2;dbl}  \\
\end{array}\right] = 
\left[ \begin{array}{ccc}
  C_{1;sd,(1)} & C_{1;sd,(1)} & C_{1;tri,(1)}  \\
   \midrule
  C_{2;sd} & C_{2;sd} & 0 \\
        0 & 0 & C_{2;tri} \\
        0 & \bar{1} & \bar{1} \\
\end{array}\right].
\end{equation}

Using this mapping on self-dual CSS codes obtained from some of the best known self-dual classical codes, we obtain the family of triorthogonal codes mentioned in Table~\ref{tab:triortho}. We now check that the resulting code is indeed triorthogonal with the claimed distance.

The weight of any non-trivial element-wise product of two rows can be evaluated as one of the following:
\begin{enumerate}
    \item $2|x\star y| = 0\mod{2}$ where $x\in C_{1;sd,(1)},y\in C_{2;sd}$.
    \item $2|x\star y| = 0 \mod{2}$ where $x,y\in C_{2;sd}$.
    \item $|x\star y| = 0 \mod{2}$ where $x \in C_{1;tri,(1)},y\in C_{2;tri}$ due to triorthogonality of $C_{1;tri}$.
    \item $|x\star y| = 0 \mod{2}$ where $x,y \in C_{2;tri}$ due to triorthogonality of $C_{1;tri}$.
    \item $|x| + |y| = 0 \mod{2}$ where $x \in C_{1;sd,(1)},y \in C_{1;tri,(1)}$ since $|x|=|y|=1\mod{2}$.
    
\end{enumerate}

 The weight of any non-trivial element-wise product of three rows can be evaluated as one of the following:
\begin{enumerate}
    \item $2|x\star y \star z|=0\mod{2}$ where $x\in C_{1;sd,(1)},y,z\in C_{2;sd}$.
    \item $2|x\star y \star z|=0\mod{2}$ where $x,y,z\in C_{2;sd}$.
    \item $|x\star y \star z|=0\mod{2}$ where $x\in C_{1;tri,(1)},y,z\in C_{2;tri}$ due to triorthogonality of $C_{1;tri}$.
    \item  $|x\star y \star z|=0\mod{2}$ where $x,y,z\in C_{2;tri}$ due to triorthogonality of $C_{1;tri}$.
    \item $|\bar{1}\star x \star y|= |x\star y| = 0\mod{2}$ where $x\in C_{1;sd,(1)},y\in C_{2;sd}$ since \Z-logicals (same as the \X-logicals) are orthogonal to the \X-stabilizers of the self-dual code.
    \item $|\bar{1}\star x \star y|= |x\star y| = 0\mod{2}$ where $x\in C_{1;tri,(1)},y\in C_{2;tri}$ due to triorthogonality of $C_{1;tri}$.
\end{enumerate}

Thus, the doubled code is indeed triorthogonal. 

We now compute the distance of the doubled code.
We can always choose $C_{1;dbl,(1)}= \begin{bmatrix}\bar{1}\end{bmatrix}$ because both the self-dual and the triorthogonal code admit strongly transversal logical \X-gates. Thus, for any $code \in \{sd,tri,dbl\}$, the code $C^{\perp}_{1;code}$ corresponding to \Z-stabilizers only contains even vectors, while the binary representation of any non-trivial \Z-logical gate (i.e. non-identity action on the codespace) has to be odd-weighted to have non-zero overlap ($\textrm{mod } 2$) with $\bar{1}$ (which represents the strongly transversal \X-logical gate).

Being a self-dual matrix, $C_{2;code}\subseteq C^{\perp}_{1;code}$, and the \Z-logical can always be chosen to be the same as the \X-logical. Thus, the distance of the CSS code is the minimum weight of any non-trivial \Z-logical.

Let the binary representation of the minimum weight non-trivial \Z-logical gates of the self-dual and the triorthogonal code be $y_{sd}$ and $y_{tri}$, respectively. It is clear that $|y_{sd}|=d_{sd}$ and $|y_{tri}|=d_{tri}$. We separate the columns of $C_{1;dbl}$ into three sectors: 
\begin{equation}
\begin{bNiceMatrix}[first-row,nullify-dots]
A & B & C \\ 
 \overbrace{C_{1;sd,(1)}} & \overbrace{C_{1;sd,(1)}} & \overbrace{C_{1;tri,(1)}} \\
 \vdots & \vdots & \vdots \\
\end{bNiceMatrix},
\end{equation}
defined by $A=\{i ; i\leq n_{sd}\}$, $B=n_{sd}+A$, and $C=2n_{sd}+\{i ; i\leq n_{tri}\}$  where $n+X=\{n+j\,|\,j\in X\}$.
Consider the vectors $\tilde{y}_{sd}=y_{sd}\oplus 0_B\oplus 0_C$ and $\tilde{y}_{tri}=(1_j + 0_A)\oplus(1_j + 0_B)\oplus y_{tri}$ for any $j\leq n_{sd}$. Clearly, both $\tilde{y}_{sd}$ and $\tilde{y}_{tri}$ are odd-weight vectors orthogonal to $C_{2;dbl}$. Thus, they both correspond to non-trivial \Z-logical gates, providing the lower bound $d_{dbl}\leq \min\{d_{sd},d_{tri}+2\}$.

To prove the upper bound, we observe that the binary representation $z=z_A\oplus z_B \oplus z_C$ of any \Z-logical satisfies
\begin{eqnarray}
    z_A + z_B &\in& C_{2;sd}^\perp, \label{eqn:cond1_dbl_Zlog}\\
    z_C &\in& C_{2;tri}^\perp,\label{eqn:cond2_dbl_Zlog} \\
    |z_B|+|z_C| &=& 0\mod{2}.\label{eqn:cond3_dbl_Zlog}
\end{eqnarray}
Since $z$ is odd-weighted and satisfies Eqn.~\ref{eqn:cond3_dbl_Zlog}, $z_A$ has to be odd-weighted. This leaves two cases:
\begin{enumerate}
    \item $|z_C|,|z_B|=0\mod{2}$: From Eqn.~\ref{eqn:cond1_dbl_Zlog}, $z_A+z_B$ is a non-trivial \Z-logical of the self-dual code. Hence, $|z|\geq |z_A|+|z_B|\geq |z_A+z_B|\geq d_{sd}$.
    \item $|z_C|,|z_B|=1\mod{2}$: From Eqn.~\ref{eqn:cond2_dbl_Zlog}, $z_C$ is a non-trivial \Z-logical gate of the triorthogonal code. Hence, $|z_C|\geq d_{tri}$ and since $z_A$ and $z_B$ are both odd-weighted, $|z|\geq d_{tri}+2$.
\end{enumerate}
This proves the necessary upper bound and hence, we conclude that $d_{dbl}=\min(d_{sd},d_{tri}+2)$.

Importantly, the same mapping can also take certain classes of DE* codes to TE* codes. We provide below a sufficient condition on the input codes in the doubling mapping to result in a TE* code.

\begin{theorem}\label{thm:dbling_te}
    The doubled code constructed from 
    \begin{enumerate}
        \item a self-dual DE* code with parameters $[[n_{de},1,d_{de}]]$ (with odd $n_{de}$) and $|M_{de}^+|-|M_{de}^-|=m \mod{8}$; and
        \item a TE* code with parameters $[[n_{te},1,d_{te}]]$ admitting strongly transversal logical \X-gate and $|M_{te}^+|-|M_{te}^-|=m \mod{8}$
    \end{enumerate}  
    is a TE* code with parameters $[[2n_{de}+n_{te},1,\min(d_{de},d_{te}+2)]]$ and $|M_{dbl}^+|-|M_{dbl}^-|=m \mod{8}$.
\end{theorem}
\begin{proof}
    For the resulting doubled code defined by $CSS(X,C_{2;dbl};Z,C_{1;dbl}^\perp)$, we need to construct $M_{te}^+$ and $M_{te}^-$ such that the code $C_{2;dbl}$ is weakly $8$-divisible (definition~\ref{def:gen_delta_div}). The generating matrix of the code 
    \begin{equation}
    C_{2;dbl}= \begin{bmatrix}
                    C_{2;de} &  C_{2;de} &  0 \\
                    0 & 0 &  C_{2;te} \\
                    0 & \bar{1} & \bar{1}
                \end{bmatrix}
    \end{equation}has $(2n_{de}+n_{te})$ columns. We choose
    \begin{eqnarray}
        M_{dbl}^+&=&M_{de}^+ \cup (n_{de}+M_{de}^+)\cup(2n_{de}+M_{te}^-) \\M_{dbl}^-&=&M_{de}^- \cup (n_{de}+M_{de}^-)\cup(2n_{de}+M_{te}^+)
    \end{eqnarray}
    where $n+M_{code}=\{n+j\,|\,j\in M_{code}\}$. The three naturally partitioned row sectors of the matrix $C_{2,dbl}$ can be seen to be TE* with this selection of $M_{dbl}^{\pm}$ as follows:
    \begin{itemize}
        \item In the first sector $[C_{2;de}\quad C_{2;de} \quad 0]$, each row $x = [x_{de}\quad x_{de} \quad \bar{0}]$ for some $x_{de}\in C_{2;de}$ satisfies 
        \begin{equation}
        \begin{split}
        |x \cap M_{dbl}^+|-|x \cap M_{dbl}^-|&=2|x_{de}\cap M_{de}^+|-2|x_{de}\cap M_{de}^-|\\
        &=0\mod{8}.
        \end{split}
        \end{equation}
        \item In the second sector $[0\quad 0 \quad C_{2;te}]$, each row $x = [\bar{0}\quad \bar{0} \quad x_{te}]$ for some $x_{te}\in C_{2,te}$ satisfies 
       \begin{equation}
       \begin{split}
        |x \cap M_{dbl}^+|-|x \cap M_{dbl}^-|&=|x_{te}\cap M_{te}^-|-|x_{te}\cap M_{de}^+|\\
        &=0\mod{8}.
        \end{split}
        \end{equation}        
        \item The last row $x=[0\quad \bar{1} \quad \bar{1}]$ satisfies
        \begin{equation}
        \begin{split}
        |x \cap M_{dbl}^+|-|x \cap M_{dbl}^-|
        &=\left(|M_{de}^+|-|M_{de}^-|\right)-\\
        &\phantom{{}=}\left(|M_{te}^+|-|M_{te}^-|\right)\\
        &=m-m = 0\mod{8}.
        \end{split}
        \end{equation}
        Hence, the resulting code $C_{2;dbl}$ is TE*. Finally, 
        \begin{equation}
        \begin{split}
            |M_{dbl}^+|-|M_{dbl}^-|&=2\left(|M_{de}^+|-|M_{de}^-|\right)-\\
            &\phantom{{}=}\left(|M_{te}^+|-|M_{te}^-|\right)\\
            &= 2(m)-(m) = m\mod{8}.
        \end{split}
        \end{equation}
    \end{itemize}
\end{proof}
The doubled color codes~\cite{bravyi_doubled_2015} are an example of this construction with $m=1 \mod{8}$. 

Every DE code is a DE* code with $|M_{de}^+|-|M_{de}^-|=|[n_{de}]|-|\emptyset|=n_{de}$. Hence, one can restrict the DE* code in Theorem~\ref{thm:dbling_te} to a DE code with $n_{de}=m\mod{8}$ to obtain a TE* code with the same properties.
This simplifies the condition of Theorem~\ref{thm:dbling_te} on the existence of some sets of indices to just the dimension of the DE code involved and will be useful in selecting appropriate DE codes to to be used for doubling, as we shall see in the next section.

\section{Quadratic Residue CSS codes}
\label{sec:qr_codes}
Classical Quadratic Residue (QR) codes (not to be confused with the graphical Quick-Response (QR) codes~\cite{tiwari_introduction_2016}) are \eczoohref[cyclic codes]{binary_cyclic} which exist for every prime length $n=\pm 1\mod{8}$. These are well studied due to their high distance, efficient decoding and finite rate~\cite{macwilliams_theory_1977}. 
When augmented by a parity bit, QR codes of length $n=-1\mod{8}$ yield self-dual DE codes of length $8j$ for some $j\in\mathbb{Z}$~\cite{ward_weight_1990} whose distances are bounded by the following theorem.
\begin{theorem}\label{thm:sd_distancebound}\cite{macwilliams_theory_1977, rains_self-dual_2002}
 (Upper bound) Any classical DE code of length $n$ has distance $d\leq 4\lfloor n/24 \rfloor +4$. (Lower bound) If the code is also an extended QR code, then $d^2-3d+4\geq n$.
\end{theorem}

A classical DE code is called \textit{extremal} if it saturates the above upper bound, and \textit{optimal} if it has the best known distance for its length.
Optimal codes might not saturate the distance upper bound, but every extremal code is optimal.

\begin{conjecture}
    CSS codes constructed from classically optimal DE codes subsume the smallest stabilizer codes with transversal Clifford group for their respective distances.
\end{conjecture}

Different extremal codes can also be compared in terms of their length, and
the first two members of the extended QR code family are extremal codes that achieve the smallest possible lengths for their respective distances. 
They yield the $[[7,1,3]]$ \eczoohref[Steane code]{steane} and the $[[23,1,7]]$ \eczoohref[Golay code]{qubit_golay} using our construction and have previously been used to construct the doubled color-code family and the $[[95,1,7]]$ TE code~\cite{bravyi_doubled_2015,sullivan_code_2024}. 
The next example of the shortest extremal code for its distance is the \eczoohref[distance-twelve extended QR code]{self_dual_48_24_12}, which corresponds to a $[[47,1,11]]$ DE CSS code. 

Other extremal classical codes that achieve the smallest possible lengths for their respective distances are not known, but many extended QR codes
are either extremal or 
close to it (see Table~\ref{tab:triply_even_from_QR}).
We use these to construct our DE code family in the second column of that table. Besides being the shortest codes with transversal Clifford gates for a given distance, the first few members of the corresponding DE quantum family were also noted to provide better depolarizing pseudo-threshold than the surface and Bacon-Shor codes~\cite{cross_comparative_2009}. 
Doubling these codes yields the most qubit-efficient weak triply even codes known for their respective distances.
They are listed in Table~\ref{tab:triply_even_from_QR}.

Infinite families of classical QR codes exist for arbitrary prime lengths $n=8j-1$ with $j\geq j_0$ for any $j_0$, as guaranteed by Dirichlet's theorem~\cite{apostol_introduction_1976}, and are straightforward to construct~\cite{macwilliams_theory_1977}. It then follows from the lower bound of Theorem~\ref{thm:sd_distancebound} and Lemma~\ref{lemma:Cperp-seldual} that DE quantum QR codes can be constructed for arbitrarily high distances. Making the conservative assumption that their distance is no more than one less  the distance of their classical counterpart (as is the case for all of the explicit codes in Table~\ref{tab:triply_even_from_QR}),
we can relate their distances to their length as follows, 
\begin{align}
   n&\leq d^{2}-d+1\\d&\leq4\lfloor(n+1)/24\rfloor+3.
\end{align}

We apply the doubling procedure recursively to the quantum version of these codes, with the $[[95,1,7]]$ code as the initial triply even code. All the doubled codes that result are TE* as a consequence of Theorem~\ref{thm:dbling_te} with $m=-1$ for each doubly even QR CSS code and the $[[95,1,7]]$ TE code. Since $m=-1$ for this TE* family, we get a transversal implementation of the logical \T-gate for each of its member.

\begin{figure}[t]
  \centering
  \includegraphics[width=0.5\textwidth]{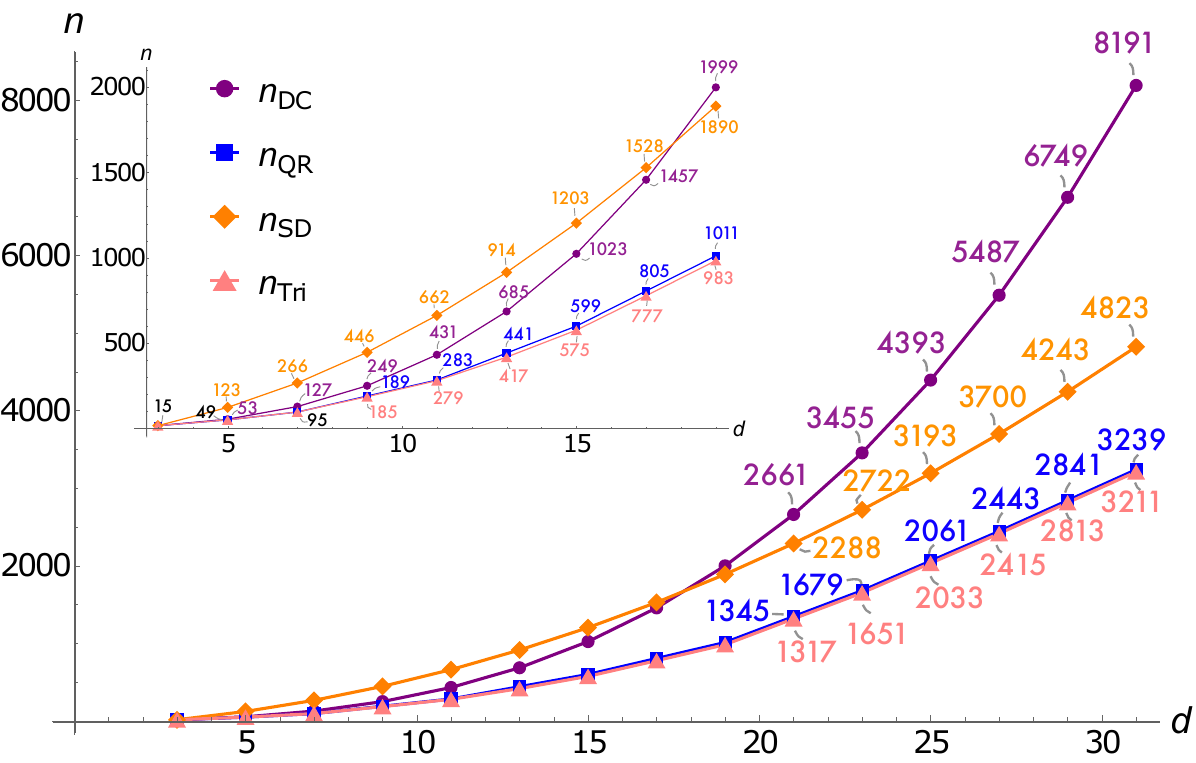}
  \caption{Comparison of physical qubit count, or length \(n\), required by TE* doubled color codes (marked ``DC'', purple), TE* codes stemming from extended QR codes (``QR'', blue) (also listed in Table~\ref{tab:triply_even_from_QR}), triorthogonal codes based on self-dual codes (``Tri'', pink) (also listed in Table~\ref{tab:triortho}) and triorthogonal codes based on a generic classical self-dual code family with guaranteed asymptotic existence (``SD'', orange)~\cite{rains_self-dual_2002}. This asymptotic family provides an upper bound for the family of triorthogonal codes and is expected to also upper bound the TE* code family.
  } 
  \label{fig:qubit-count-comparison}
\end{figure}

Infinite families of self-dual codes also exist whose distance scales linearly with the number of bits~\cite{rains_self-dual_2002, krasikov_linear_1997}. Applying the doubling procedure recursively to the quantum version of such a family produces triorthogonal codes with parameters $[[O(d^2),1,d]]$, whose length scaling with \(d\) matches that of the family alluded to in a previous work~\cite[Section VI]{haah_towers_2018}. 

We consider the best known examples of self-dual classical codes and double them to construct triorthogonal codes which have the least physical qubit overhead for their given distances. We list examples of this family in Table~\ref{tab:triortho} and plot their parameters in Figure~\ref{fig:qubit-count-comparison} alongside the constructed TE* codes, the asymptotic family of doubled self-dual codes and the color codes, whose qubit count grows as the third power of the distance.

We note here that infinite families of self-dual codes are guaranteed to exist with $d/n > \delta$ for each member where $\delta \sim 0.11$~\cite{rains_self-dual_2002,krasikov_linear_1997}. Both our triorthogonal and TE* families outperform the codes obtained by doubling the asymptotically guaranteed self-dual codes for the distances we could compute (Figure~\ref{fig:qubit-count-comparison}). This advantage will be retained at higher distances for the triorthogonal family due to the choice of self-dual codes optimized for the lowest qubit count.
It is also expected to be retained for the TE* family due to the additional structure present in the QR codes. Consequently, the scaling of this asymptotically guaranteed family places an upper bound on the physical qubit overhead of any triorthogonal code family and is expected to also upper bound the TE* family.

While we gain significant improvements in the physical qubit overheads, the stabilizer generators of our TE* and triorthogonal code families are geometrically non-local and their weights grow linearly with the distance of the code (equivalently, as square root of the code length). This is in contrast with the doubled color code family where most stabilizer generators are geometrically local and their weights can be made constant in the distance as a result of a weight reduction procedure~\cite{bravyi_doubled_2015}. The same weight reduction procedure can be employed for our codes but due to the underlying QR CSS codes requiring high weight stabilizers (as opposed to the 2D color codes with constant weight stabilizer generators), the weight reduced stabilizer generator set still admits some operators with weights growing linearly with the distance.

\section{Conclusion}\label{sec:conclusion}
    
Our weak triply even and triorthogonal families can be utilized as a component of magic state distillation protocols~\cite{bravyi_universal_2005,reichardt_quantum_2005,reichardt_error-detection-based_2009,reichardt_quantum_2009,campbell_roads_2017}. These codes are expected to perform better than the general self-dual code based family which, based on their parameters, exhibit a distillation cost scaling exponent $\gamma=\log_d(n/k)\rightarrow 2$ in the asymptotic limit. However, the exact scaling of their qubit overhead with distance (and thus, the value of $\gamma$) remains an open question. 

The constructed weak triply even codes are also compatible with the code conversion protocol of~\cite{sullivan_code_2024} which allows universal fault-tolerant computation without distillation. 

Extremal self-dual codes only allow CSS code distances of the form $d=4m+3$ for $m\in\mathbb{Z}$. It would be interesting to look for doubly even CSS codes with distances $4m+1$ that are shorter than quantum QR codes with distances $4m+3$ for the same $m$. Such codes, in conjunction with the QR codes, would immediately yield shorter weak triply even codes under the doubling map. 
The sole notable example used here is the $[[17,1,5]]$ \eczoohref[color code]{488_color}~\cite{bombin_topological_2006} based on the square-octagon lattice. It is a doubly even, distance five degenerate code which fills in the gap between distance three and seven quantum QR codes.
We leave filling the rest of the gaps as an exciting open question for future work.

Triorthogonal codes and CSS-T are well understood for the case of multiple logical qubits. However, all the mentioned results about transversal gates for the divisible classes only hold for the $k=1$ case. It is an interesting open question to characterize divisible codes and their transversal gates for case of $k>1$ to allow for constructions with potentially non-vanishing rates.

\section{Acknowledgements}
The authors thank Alexander Barg, Narayanan Rengaswamy,
Ben Brown,
Kenneth Brown,
Prakhar Gupta,
Vasanth Pidaparthy,
and Michael Vasmer
for insightful discussions. We especially thank Anqi Gong for their question on \href{https://scirate.com/}{SciRate} which motivated us to study the code families more deeply. We also thank the Simons Institute for their \href{https://simons.berkeley.edu/workshops/application-driven-coding-theory}{public repository of lectures}~\cite{simons_institute_for_the_theory_of_computing_application-driven_2024} which helped us understand the preliminaries for this work.
This work is supported in part by NSF grants OMA-2120757 (QLCI) and CIF-2330909. Part of this research was performed while SPJ was
visiting the Institute for Mathematical and Statistical
Innovation (IMSI), which is supported by the National
Science Foundation (Grant No. DMS-1929348).

\bibliographystyle{ieeetran}
\bibliography{ieeeabrv,references}

\end{document}